\documentclass[12pt]{article}
\usepackage{TpackTron}
\title{\tron, a combinatorial Game on abstract Graphs}
\author{Tillmann Miltzow}

\begin{document}
\maketitle

\begin{abstract} 
	We study 
	the combinatorial two-player
	game \tron. We answer the 
	extremal question on general
	graphs and also consider smaller
	graph classes. 
	Bodlaender and Kloks conjectured in \cite{TronOnTrees}
	PSPACE-completeness.
	We proof this conjecture.
\end{abstract}

\section{Introduction}

	The movie Tron from 1982 inspired the computer game Tron \cite{TronWebseite}.
	The game is played in a rectangle
	by two light cycles or 
	motorbikes, which try to cut each other off,
	so that one, eventually has  to hit a wall
	or a light ray.
	We consider a natural abstraction of the game, 
	which we define as follows:
	Given an undirected graph $G$, two opponents 
	play in turns. The first player (\alice )
	begins by picking a start-vertex of $G$ and
	the second player (\bob ) proceeds 
	by picking a different start-vertex for himself. 
	Now \alice\ and \bob\ take turns, by moving 
	to an adjacent vertex from their respective previous one 
	in each step.
	While doing that it is forbidden to
	reuse a vertex, which was already traversed
	by either player.
	The game ends when both players
	cannot move anymore.
	The competitor who traversed more
	vertices wins.
	Tron can be pictured with
	two snakes, which eat up pieces of
	a tray of cake, with the restriction that
	each snake only eats adjacent pieces and starves 
	if there is no adjacent piece left for her.
	We assume that both contestants 
	have complete information at all times.
	
	
	Bodlaender is the one who first introduced 
	the game to the science community, and according to him 
	Marinus Veldhorst proposed to study \tron.
	Bodlaender showed PSPACE-completeness for \emph{directed}
	graphs with and without given start positions \cite{DBLP:journals/tcs/Bodlaender93}. 
	Later, Bodlaender and Kloks showed that
	there are fast algorithms for \tron\ on trees \cite{TronOnTrees}
	and NP-hardness and co-NP-hardness for \emph{undirected} graphs.

	We have two kind of results. On the one hand
	we investigated by how much \alice\ or 
	\bob\ can win at most. It turns out, that
	both players can gather all the vertices
	except a constant number in particular
	graphs. This results still holds for
	$k$-connected graphs. For planar
	graphs, we achieve a weaker,
	but similar result.
	We also investigated the 
	computational complexity question.
	we showed PSPACE-completeness for
	\tron\ played on \emph{undirected} graphs
	both when starting positions are given
	and when they are not given.

	
	Many proofs require some tedious case analysis. 
	We therefore believe that thinking about the cases 
	before reading all the details will
	facilitate the process of understanding.
	To simplify matters, we neglected constants
	whenever possible.
	
\section{Basic Observations}\label{basObs}
	The aim of this section is to show some
	basic characteristics and introduce some notation, 
	so that the reader has the opportunity to become familiar 
	with the game.
	
	\begin{definition}
		Let $G$ be a graph, and \alice\ 
		and \bob\ play one game of \tron \ 
		on $G$. Then we denote with
		$\alicewin$ the number of vertices
		\alice\ traversed and with $\bobwin$
		the number of vertices \bob \ traversed
		on $G$. The outcome of the game 
		is $\bobwin / \alicewin$.
		We say \bob\ wins iff $\bobwin > \alicewin$,
		\alice\ wins iff $\alicewin > \bobwin$
		and otherwise we call the game a tie.
		We say \bob\ plays \emph{rationally}, if 
		his strategy maximizes the outcome
		and we say \alice\ plays \emph{rationally}
		if her strategy minimizes the outcome.
		We further assume that \alice\ and \bob\ 
		play always rational.
	\end{definition}
	
		Here we differ slightly from \cite{DBLP:journals/tcs/Bodlaender93},
		where \alice\ loses if both players receive
		the same amount of vertices.
		We introduce this technical nuance,
		because it makes more sense in regard of 
		the extremal question and is
		not relevant for the complexity question.
	
		Now when you play a few games of \tron\ on a "random"
		graph, you will observe that you will usually end up in 
		a tie or you will find that one of the players
		made a mistake during the game.
		So a natural first 
		question to ask is if \alice\ or
		\bob\ can win by more than one at all.
	
	\begin{example}[two paths]
		\label{twopath}
		Let $G$ be a graph which consists of two paths
		of  length $n$. 
		On the one hand, \alice\ could start close to the middle of
		one of the paths, then \bob\ starts at
		the beginning of the other path, and thus
		wins. On the other hand if \alice\ 
		tries to start closer to the
		end of a path \bob\ will cut her
		off by starting next to
		her on the longer side of her path.
		The optimal solution
		lies somewhere in between and 
		a bit of arithmetic reveals that
		for the optimal choice $\bobwin / \alicewin$ tends to $2/(\sqrt{5}-1)$
		as $n$ tends to infinity.
	\end{example}
	
	And what about \alice ? We will modify 
	our graph above by adding a super-vertex $v$
	 adjacent to every vertex of
	$G$. Now when \alice\ starts there
	the first vertex on $G$ will be taken
	by \bob\ and \alice\ will take the second
	vertex on $G$. So we see that the roles of
	\alice \ and \bob\ have interchanged. 
	
	\begin{lemma}[Super-vertex]\label{lem:supervertex}
		Let $G$ be a graph where 
		$\left(\bobwin\right/ \alicewin)_G > 1 $ and
		$F$ be the graph we obtain by
		adding a super-vertex $v$ adjacent to
		every vertex of $G$.
		It follows that
		$$\left(\frac{\alicewin}{\bobwin }\right)_{F} \geq  \left(\frac{\bobwin }{\alicewin}\right)_{G}.$$
	\end{lemma}
	
	So lemma \ref{lem:supervertex} simplifies matters.
	Once we have found a good graph for 
	\bob\ we have automatically found a good graph
	for \alice. But the other direction
	holds as well. Let $G$ be a graph where
	\alice\ wins and let us say she starts at
	vertex $v$. Delete vertex $v$ from $G$ to 
	attain $H$. Now the situation in
	\alice's first move in $H$ is the same as 
	\bob's first move in $G$. And \bob's first
	move in $H$ includes the options \alice\ had
	in her second move in $G$.
	\begin{lemma}
		Let $G$ be a graph where $\left(\bobwin\right/ \alicewin)_G < 1$
		and $H$ be the graph we obtain by deleting
		the vertex where \alice\ starts. It follows that
		$$\left(\frac{\bobwin + 1}{\alicewin}\right)_{H} \geq  \left(\frac{\alicewin}{\bobwin }\right)_{G}.$$
	\end{lemma}
		Note that the starting vertex of \alice\
		need not be unique, even when
		\alice\ wins.
		To see this consider the complete
		graph with an odd number of vertices.
	
		Bodlaender and Kloks \cite{TronOnTrees} 
		showed the first equation of lemma \ref{lem:tree}
		in theorem 3.1.
	
	\begin{lemma}[Trees]\label{lem:tree}
		Let $T$ be a tree then $\alicewin \leq \bobwin + 1$ 
		and $\bobwin   \leq 2 \cdot \alicewin$
	\end{lemma}
	\begin{proof}
			The idea of the proof is
			to describe a strategy for \alice\ and \bob\ explicitly.
			Let $v$ denote the starting vertex of Alice with 
			$w_1,\ldots, w_k$ its neighbors and
			$l_i$ the length of the longest path from
			$w_i$ in $T-\{v\}$ for all $i=1,\ldots,k$.
			Further, we denote with $j$ an index which
			satisfies $l_j = \max_{i=1,\ldots,k}{l_i}$.
			If \bob \ chooses $w_j$ as start-vertex and
			thus obtains at least $l_j$, while \alice\ receives
			at most $l_j + 1$.
			
			For the second inequality, let \alice\ start in the middle
			of a longest path. Thus she divides the tree into
			smaller trees $T_1,\ldots,T_k$. \bob\ chooses
			one of them and receives at most as many vertices 
			as the length of longest path in $T$.
			\alice\ will enter a different tree,
			which contains one half of the longest path.
			So she receives at least half of the longest
			path.
	\end{proof}

\section{Extremal Question}\label{extrQuest}
	In this section we want to answer the extremal question
	for \tron. That is: Is there a non-trivial
	upper bound on $\bobwin/\alicewin$ as a function 
	of the number of vertices of $G$?
	The answer is surprisingly no. For every natural number $n$ exists
	a graph with $n$ vertices, such that $\bobwin = n-c$ for some constant $c$.
	The idea behind it is fairly easy. Think of two
	motorcycles on a highway. 
	If one of the motorcycles pushes the other from 
	the highway to the smaller and slower country roads, 
	it can encircle the other using the highway and 
	traverse the rest of the map comfortably alone.
	
	It seems convenient to study a simpler
	example first.
	
	\begin{example}[big-circle] \label{bigcircle}
	In this example we consider a cycle 
	of length $n$, and a subtle change of the 
	rules. We assume, that \alice\ has to make
	two moves before \bob\ joins the game.
	Now the analysis of this example is short:
	\alice\ decides for a vertex and a direction
	and \bob\ can simply start in front of her
	and take the rest of the circle.
	This example will also work, if only every $100^{th}$
	vertex is   an admissible start-vertex for \bob .
	\end{example}
	
	\begin{example}[visage]
		This example consists of three parts:
		an overhead graph, a big-circle and a bottleneck
		as depicted on the left-hand side of figure \ref{fig:VisageTillTorsten}. 
		The overhead graph can be any graph where \bob\ 
		wins. The two paths in example \ref{twopath}
		give us such a graph. It suffices to take paths of length
		$3$ for our purpose. 
		The big-circle consists of a large cycle
		of length $4*l$. The last part is a bottleneck
		between the first two parts and consists of two
		singular vertices. The bottleneck is
		adjacent to every vertex of the overhead graph
		but only to every fourth vertex 
		of the cycle. Alternating between the
		two bottleneck vertices.
		\begin{figure}[h]
			\begin{center}
				\includegraphics[width = 0.9\textwidth]{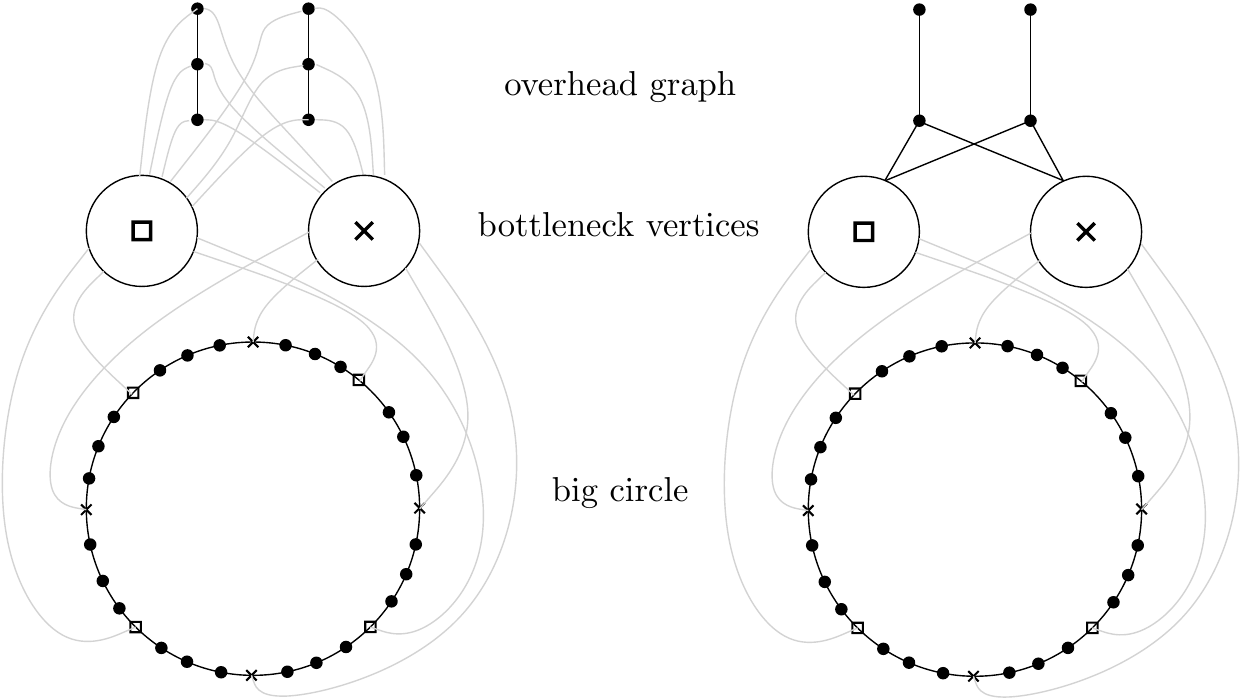}
				\caption{The ordinary visage on the left-hand side and the visage from Torsten Ueckerdt
				on the right-hand side.}
				\label{fig:VisageTillTorsten}
			\end{center}
		\end{figure}
		\bob\ has a strategy to 
		gather $n-c$ vertices,
		where $n$ denotes the total
		number of vertices in the visage 
		and $c$ some constant.
	\end{example}
	\begin{proof}
		We will give a strategy for \bob\
		for all possible moves of \alice.
		
		\noindent
		\texttt{Case1}
				\alice\ starts in the 
				overhead graph. 
				\bob\ will then also start in the overhead graph
				and win within the overhead graph. 
				So \alice\ has to leave  
				the overhead graph eventually
				and go to one of the
				bottleneck-vertices. 
				\bob\ waits one more turn 
				within the overhead graph.
				If \alice\ tries to go 
				back to the overhead graph,
				\bob\ will go to the other bottleneck-vertex
				and trap her.
				Thus \alice\ will have to go to the
				big-circle and once there she will
				have made already
				two turns when \bob\ enters the
				big-circle. We already studied this situation
				in example \ref{bigcircle}.
				
			\noindent
			\texttt{Case2}
				\alice\ starts in one of the bottleneck
				vertices. \bob\ will then again
				start somewhere in the overhead 
				graph. The situation is as in case 1.
			
			\noindent		
			\texttt{Case3}
				At last we consider the case where \alice\
				starts in the big-circle. In this case, 
				\bob\ will start 
				on the closest bottleneck-vertex to \alice\
				and then quickly go to the other bottleneck-vertex via 
				the overhead graph. Thus she cannot
				leave the big-circle. Finally  he enters the
				big-circle and cuts her off.
	\end{proof}

		\begin{example}[Torsten-visage]
		Torsten Ueckerdt showed
		with a very similar construction how 
		to reduce $c$ to $8$.
		As depicted on the right hand side of figure 
		\ref{fig:VisageTillTorsten}.
		Here not every vertex
		of the overhead graph is connected
		to the bottleneck.
		We want to point out, that
		the crossing number is
		two, so the graph is almost planar.
		We omit the proof as it does not
		involve any new ideas.
	\end{example}

		In addition, lemma \ref{lem:supervertex} gives us a graph
		where \alice\ can obtain all vertices except a constant
		amount.
		The natural question is, for which graph classes
		is this kind of construction possible. 
		Which graph classes should we consider?
		
		Very interesting is always the
		planar case, because planar graphs are very well
		studied and very close to the original game. 
		We will show that we have graphs
		with $\bobwin/\alicewin = \Theta(\sqrt{n})$ and 
		graphs with $\alicewin/\bobwin = \Theta(\sqrt{n})$.
		
		In the visage, \alice\'s first move
		on the big-circle gives her a direction and she
		cannot reconsider. This was the key
		for \bob\ to be able to cut her off. In a highly connected graph,
		we expect intuitively that \alice\ has
		enough freedom to avoid to get imprisoned.
		This motivates us to study $k$-connected graphs.
		Surprisingly, we are able to
		adapt the 
		visage so that it becomes 
		arbitrarily highly connected.

		\subsection{Planar Graphs}
		
		We construct a planar visage. Again
		it will be more convenient to
		study a simpler example first.
		\begin{example}[long-path]\label{ex:longPath}
			We consider now a path of length $n$, again
			with the subtle change of rules, that \alice\ has
			to make two moves before \bob\ starts.
			Now let us say, that \alice\ goes
			along the path and has $x$ vertices she might 
			be able to reach.
			Well, if \bob\ cuts her off 
			the outcome will be 
			$\bobwin/\alicewin = \Theta(x)$.
			So \alice\ could choose $x$ to be fairly small.
			But then \bob\ would just choose the other
			side and the outcome would become
			$\bobwin/\alicewin = \Theta(n/x)$.
			This observation implies
			that \alice\ best choice is to choose $x = \Theta(\sqrt{n})$ and
			thus the value of the game is $\Theta(\sqrt{n})$.
		\end{example}
		
			\begin{figure}[h]
				\begin{center}
					\includegraphics[width = 0.9\textwidth]{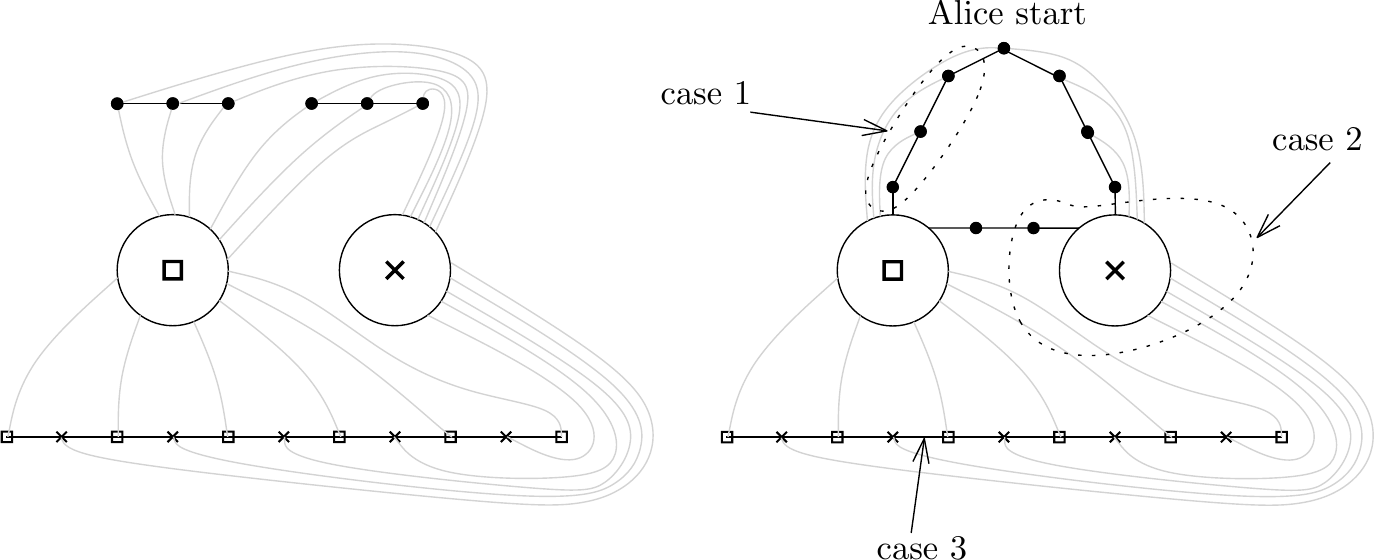}
					\caption{planar visage for \bob }
					\label{fig:planarVisageAliceBob}					
				\end{center}	
			\end{figure}
		
		The reader might already guess the planar visage
		we present now.
		\begin{example}[planar visage]
			The planar visage differs from the visage,
			only in one point, namely the big-circle
			is replaced by a long path. We can see from
			the drawing on the left-hand side of figure \ref{fig:planarVisageAliceBob},
			that this graph is planar.
			It is clear that the value is $\Theta(\sqrt{n})$
		\end{example}
		
		Unfortunately we cannot apply
		lemma \ref{lem:supervertex} to obtain
		a planar graph where \alice\ wins.
		Even if we add a super-vertex which 
		is only connected, to wisely chosen 
		vertices. Instead we construct a
		new overhead graph.
		
		\begin{example}[planar visage for \alice ]
			This example is different from the previous one
			in two ways. Obviously the overhead 
			graph has changed, as you can 
			see on the left had side of 
			figure \ref{fig:planarVisageAliceBob}, 
			but more subtly the distance 
			between vertices adjacent to the
			bottleneck increased  from $4$ to $7$.
			We claim that $\alicewin/\bobwin = \Theta(\sqrt{n})$
		\end{example}
		
	\begin{proof}
		We give an explicit strategy for
		\alice .
		\alice 's start-vertex is marked 
		in figure \ref{fig:planarVisageAliceBob}.
		
		\noindent
		\texttt{Case 1}
				\bob\ starts somewhere on a path
				from \alice 's start-vertex to a bottleneck-vertex.
				\alice\ can just go directly
				to the corresponding bottleneck-vertex and
				trap \bob\ this way.
		
		\noindent
		\texttt{Case 2} \bob\ starts w.l.o.g. on 
				the $\times$-vertex or
				on the adjacent vertex on the path between
				the two bottleneck vertices. 
				\alice\ will slowly go
				to the $\Box$-vertex, but hurry up, if
				\bob\ shows any sign that
				he wants to go in that direction as well.
		
		\noindent
		\texttt{Case 3} The only remainung case is where \bob\ starts
				on the long path.
				\alice\ will go to the bottleneck-vertex,
				which is closer to \bob\ and then 
				to the other bottleneck-vertex.
				Thus \bob\ cannot leave the long path and \alice\
				enters it second.
	\end{proof}
	
		We end the section on planar graphs
		with a remark, that in both planar visages 
		we presented, it is possible for the winning player to cut 
		off his or her opponent within at most 20 turns.
		But this might not be the optimal strategy,
		as shown in example \ref{ex:longPath}.
		Thus you cannot proof a lower bound of
		$\Omega(\sqrt{n})$, by showing that
		each player receives at least $\Omega(\sqrt{n})$ vertices.
		However, the author feels that a totally different 
		kind of idea would be needed to improve
		the planar visage.
		We conjecture that the upper 
		bound is tight.
		
\subsection{$k$-connected Visage}

	Now we show how to construct a \emph{$k$-connected visage}.
	The essence is twofold. First, we increase the number of 
	bottleneck-vertices and secondly we replace paths of the big-circle
	by double-trees which will be introduced shortly.
	Double-trees have a high-connectivity, a bottleneck and
	the vertex-set can be traversed entirely.
	
		\begin{figure}[h]
					\begin{center}
						\includegraphics[width = 0.9\textwidth]{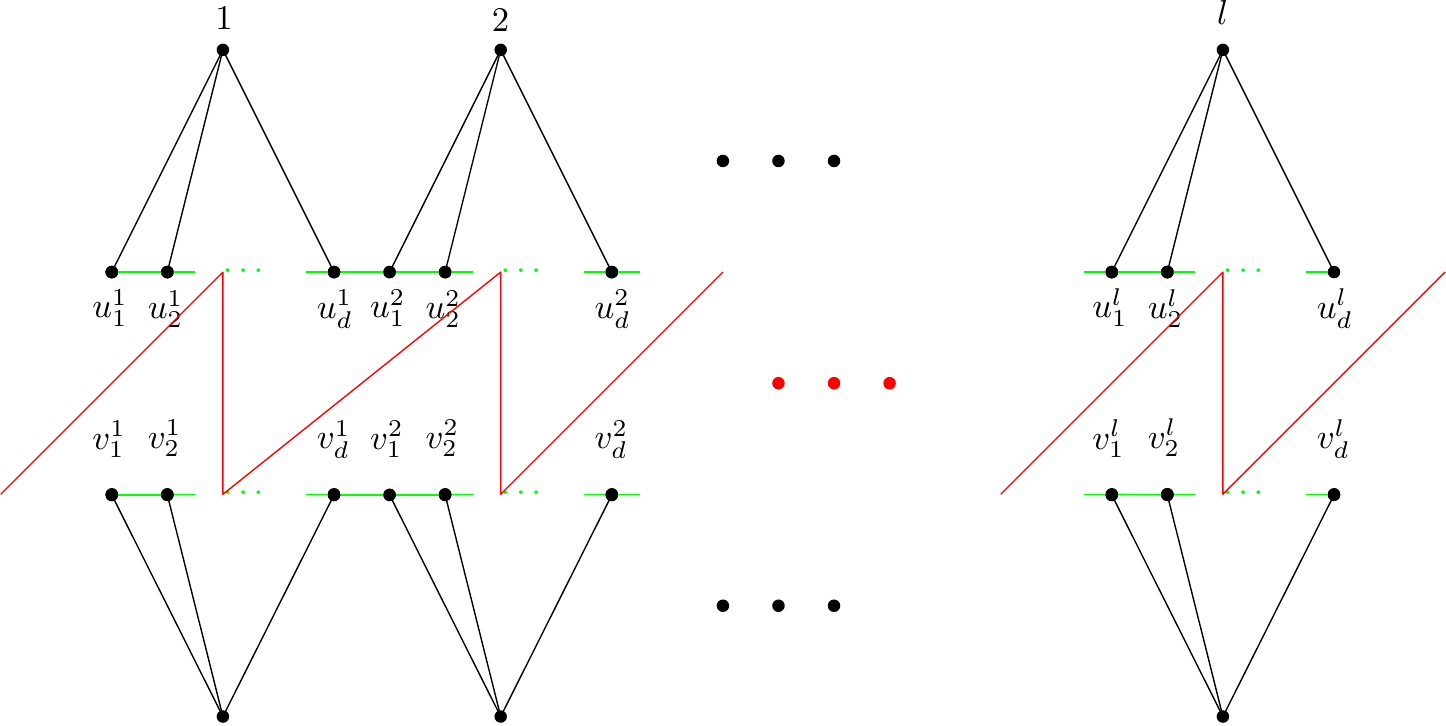}
						\caption{All the leaves with their names and parents. In green is the edge-set 
						$E_1$ indicated and in red one of the cycle partitioning the set of leaves.}
						\label{fig:doubleTreeLeafNames}
					\end{center}	
		\end{figure}
		
	\begin{example}[Double-tree]
	 	A \emph{double-tree} of degree $d$ and height $h$
		consists of two fully balanced rooted trees with 
		$d$ children for each inner vertex 
		and height $h$.
		The two double-trees are connected
		as explained below.
		
		From left to right we name the leaves of the upper half of the 
		double-tree 
		$$u_1^1, u_2^1, u_3^1, \ldots ,u_d^1 ,u_1^2, u_2^2, u_3^2, 
		\ldots, u_d^2 ,\ \ldots \ ,u_1^l, u_2^l, u_3^l, \ldots ,u_d^l$$
		and likewise we name the leaves from the
		lower half of the double-tree 
		$$v_1^1, v_2^1, v_3^1, \ldots, v_d^1, v_1^2, v_2^2, v_3^2, 
		\ldots, v_d^2, \ldots, v_1^l, v_2^l, v_3^l ,\ldots, v_d^l$$ 
		as depicted in figure \ref{fig:doubleTreeLeafNames}.
		With $l$ we denote the number of parents of the leaves.

		We add edge-set 

								$$E_1 = \left\{(u_i^j,u_{i+1}^j): i= 1,\ldots,(d-1), j=1,\ldots,l\right\} $$
								$$ \cup \left\{(u_d^j,u_{1}^{j+1}): j=1,\ldots,(l-1) \right\}$$
								$$ \left\{(v_i^j,v_{i+1}^j): i= 1,\ldots,(d-1), j=1,\ldots,l\right\} $$
								$$ \cup \left\{(v_d^j,v_{1}^{j+1}): j=1,\ldots,(l-1)\right\}$$
										
								and 
									
								$$ E_2 = \left\{(u_n^j,v_m^j): j = 1,\ldots,l ; n \leq m\right\} $$ 
								$$ \cup \left\{(u_n^j,v_m^{j+1}): j = 1,\ldots,(j-1) ; m < n \right\}  \\$$
								$$ \cup \left\{(u_n^l,v_m^{1}):  m < n \right\}.$$		 
	\end{example}
	
	\begin{figure}[h!]
		\begin{center}
			\includegraphics[width = 0.7\textwidth]{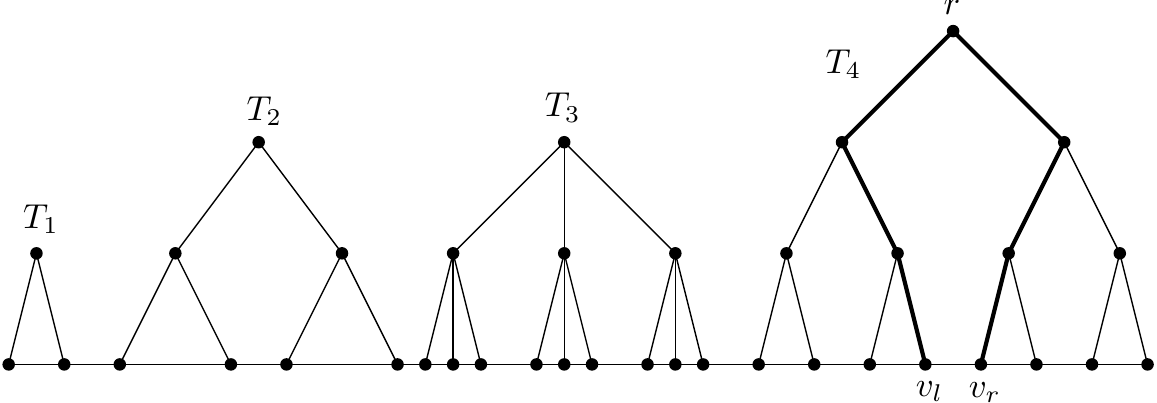}
			\caption{The induction step of lemma \ref{lem:treeTraversing}}
			\label{fig:treeTraversing}
		\end{center}	
	\end{figure}
	
	\begin{lemma} \label{lem:treeTraversing}
		Let $T_1, T_2,\ldots , T_n$ be some trees.
		And each inner-vertex has at least $2$ children.
		Now we draw all the trees crossing-free in a half plane
		s.t. all leaves are on the boundary of the half plane.
		We define leaves  as adjacent iff the line segment 
		connecting them does not contain any other vertex.
		It follows, that there is an Hamilton-path from the 
		leftmost to the rightmost leaf.
	\end{lemma}
	
	\begin{proof}
		We construct a special partition of the vertex-set 
		of the trees into paths starting and ending in adjacent leaves. 
		In a second step we just connect these
		paths canonically.
		
		We do the first step by induction. If every tree is
		just a single vertex we define the
		partition to be the collection of all one element sets.
		Now let $T_1, T_2,\ldots , T_n$ be some trees
		as described above and $r$ a root of $T_i$
		and we assume, that $T_i$ consists of more
		than one vertex. Let $v_l$ be the rightmost
		leaf of the leftmost subtree of
		$T_i$ and $v_r$ the leftmost leaf of
		the second subtree (left to right)
		of $T_i$. See figure \ref{fig:treeTraversing}.
		There exists exactly one path from
		$v_l$ to $v_r$ via $r$ within $T_i$. 
		We add this path to the partition set 
		and delete it from the trees. Thus 
		we end up with a new 
		bunch of trees with fewer vertices
		in total which we can partition by
		induction.
		
		Now we connect the right end of each path
		to the left end of the next path to get the 
		desired Hamilton-path.
	\end{proof}
	
	\begin{lemma}\label{lem:doubleTreetraversing}
		There is an Hamilton-path from
		one root of the double-tree
		to the other root.
	\end{lemma}
	
	\begin{proof}
		We start our tour at the top root and go down to $u_1^1$
		and from there to $u_2^1$. 
		Next we use the path constructed in lemma \ref{lem:treeTraversing}
		to traverse the rest of this half of the
		double-tree. After ferrying over from 
		$u_d^l$ to $v_d^l$ 
		we copy 
		our path in reverse order and
		thus have traversed every vertex without
		reusing any.
	\end{proof}

	\begin{lemma}\label{lem:doubleTreeKConnected}
		The double-tree is d-connected
	\end{lemma}
	\begin{proof}
		For any $2$ nodes we will show that 
		there are $d$ vertex-disjoint paths connecting them.
		And thus by Menger's theorem \cite[section 3.3]{diestel2005graph},
		we know that the graph is $d$-connected.
		Close Observation shows that there are $d$ vertex-disjoint
		cycles partitioning the leaves. Cycle $i$ is described via
		$u^1_i, v^1_i, u^2_i, v^2_i, u^3_i, v^3_i, \ldots , u^d_i, v^d_i$,
		see the red path in figure \ref{fig:doubleTreeLeafNames}.
		To find $d$ vertex-disjoint paths from some $w_1$
		to some $w_2$,
		it suffices to show that there exist
		$d$ vertex disjoint paths to the $d$
		different cycles. This is clear 
		for every leaf. It is also clear for every 
		inner vertex as every inner vertex has $d$ 
		children and every parent of a leaf is adjacent 
		to all $d$ cycles.
		To construct $d$ disjoint paths from
		$w_1$ to $w_2$, we use $d$ vertex-disjoint direct paths
		from $w_1$ to the cycles and from $w_2$ 
		to the cycles. These paths can be connected to 
		each other, via the cycles. The only thing that can go wrong,
		is that a path from w.l.o.g. $w_1$ to a cycle
		uses $w_2$. But this can only happen once and would give us
		a path from $w_1$ to $w_2$, 
		which is disjoint from the others.
	\end{proof}

	\begin{figure}[h!]
					\begin{center}
						\includegraphics[width = 0.4\textwidth]{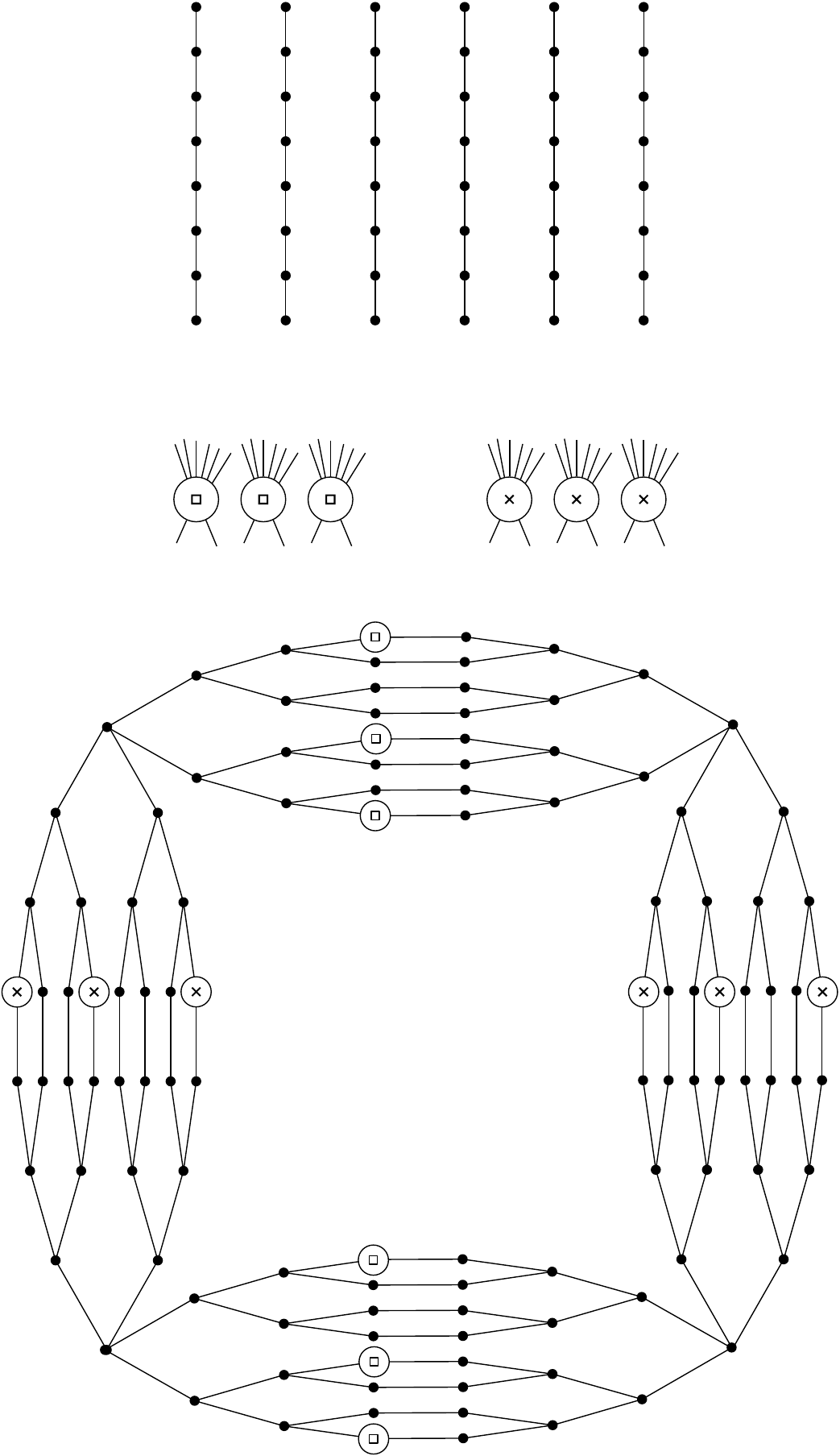}
						\caption{A $3$-connected visage with $4$ 
						double-trees, the leaves marked with a 
						$\Square$ are connected with a 
						$\Square$-bottleneck-vertex and
						likewise leaves marked with a 
						$\times$ are connected with a 
						$\times$-bottleneck-vertex.
						For clarity of the drawing
						we we chose the height of the double-trees
						too small and omitted some
						edges of the double-trees.	}
						\label{kConnectedVisage}
					\end{center}	
		\end{figure}
	
	\begin{lemma}[Many afar leaves]\label{lem:manyAfarLeaves}
		For all natural numbers $m,n$ and $d \geq 2$ there exists some
		$h_0\in \N$ such that a double-tree
		with height $h\geq h_0$ 
		and degree $d$ has at
		least $m$ leaves which all have
		pairwise distance at least $n$.
	\end{lemma}
	\begin{proof}
		Every vertex has
		at most degree $c = d+3$.
		The number of leaves grows
		strictly monotone with $h$.
		In fact, it grows exponentially.
		
		We proceed by induction on $n$.
		If we have leaf $v$,
		then we have at most $c + c^2 +c^3 + \ldots +c^m$
		many leaves within distance $m$ to $v$
		and thus we can find a leaf $w$ which
		has distance larger than $m$ to $v$,
		if we choose $h_0$ large enough.
		This shows the base case $n=2$.
		
		On the other hand $n$ leaves have at most
		$n(c + c^2 +c^3 + \ldots +c^m)$
		leaves within distance $m$ and thus
		we could find a leaf which has 
		distance $m$ to all the other 
		leaves, if $h_0$ is large enough.
		This proves the induction step.
	\end{proof}
		
	\begin{example}[$k$-connected visage]
	  The $k$-connected visage
	  is depicted in figure \ref{kConnectedVisage}.
		It consists again of three parts:
		The overhead graph is composed out 
		of $k$ sufficiently long paths.
		The bottleneck-vertices are separated 
		into two groups each of $k$ $\Square$-vertices
		and $k$ $\times$-vertices.
		Every vertex of the overhead graph
		is connected 
		to every vertex of the bottleneck.
		The big-circle consists of degree $k$ double
		trees which we string together by their roots like a
		jeweler strings pearls together in a necklace.
		Now the height $h$ of the double trees
		is picked such that each double tree has
		$k$ leaves, which all have pairwise 
		distance at least $2k$, see lemma \ref{lem:manyAfarLeaves}.
		On these spots we connect alternating either $k$
		$\Box$-vertices or $k$ $\times$-vertices.
		This construction is indeed $k$-connected and
		\bobwin\ is everything but a constant.
	\end{example}
	
	\begin{proof}
		To see $k$-connectedness 
		we assume $(k-1)$ vertices have
		been deleted. We will show
		that the graph is still connected.
		It suffice to show that, we can 
		still connect every vertex to one
		of the bottleneck-vertices,
		let us say $v$.
		This is clear for any vertex of the 
		overhead graph. Any other vertex
		of the bottleneck is connected to
		$v$ via any vertex of the 
		overhead graph.
		In lemma \ref{lem:doubleTreeKConnected} we showed, that every double-tree 
		is $k$-connected.
		Thus every vertex stil has
		a path to a leaf that is connected
		to one of the bottleneck-vertices,
		which is itself connected to $v$.
		
		To prove that \bobwin\ is everything 
		but a constant, we only point out
		where \bob s strategy differs from 
		his strategy for the ordinary visage.
		\begin{itemize}
			\item \bob\ can wait arbitrarily long in 
				the overhead graph, before he has to enter 
				the big-circle if the paths are long enough,
				but still constant length as a function of $n$.
			\item \alice\ might go back and forth between 
				the overhead graph and the bottleneck. 
				While she does that, \bob\ can maintain
				that the number of $\Box$-vertices
				equals the number of $\times$-vertices.
				Thus when \alice\ enters the big-circle via
				one kind of vertex,
				\bob\ can still enter via the other kind,
				at some much later stage.
			\item Once \alice\ enters the big-circle (w.l.o.g. via 
				a $\Square$-vertex.)
				\bob\ can traverse all remaining $\Square$-vertices in $2k-3$
				moves and thus prevent \alice\ from returning 
				to the overhead graph.
			\item After a constant number of turns in the big-circle
				\alice\ has to use a root of a double tree.
				This forces her to decide for a direction, she 
				wants to go on the big-circle. 
				Accordingly \bob\ enters the big-circle, once she
				has decided and cuts her off by reaching the next root
				on the big-circle earlier than she does.
		\end{itemize}
	\end{proof}
	
	Adding a super-vertex as in lemma 
	\ref{lem:supervertex} gives us instantly a $k$-connected 
	visage which is good for \alice .	
\section{Complexity Question}\label{complQuest}

	In this section we show that Tron is PSPACE-complete.
	To do this, it turned out to be convenient to consider
	variations where the graph is directed and/or
	start positions for \alice\ and \bob\  are given.
	We reduce Tron to quantified boolean formula(QBF). 
	It is well known, that it is PSPACE-complete
	to decide if a QBF $\varphi$ is true.
	A quantified boolean formula has the
	form $\varphi \equiv \exists x_1\forall x_2 
	\exists x_3 \forall x_4  \ldots : C_1\wedge \ldots \wedge C_k$ 
	with each $C_i \equiv L_{i_1}\vee L_{i_2}\vee L_{i_3}$ 
	and $L_{i_j}$ some Literals) \cite[section 8.3]{DBLP:books/daglib/0086373}.
	In theorem \ref{dgsPSPACE} we will construct for each $\varphi$  
	a directed graph $G_{\varphi}$
	with given start positions $v_1$ and $v_2$
	such that \alice\ has 
	a winning strategy if and only if $\varphi$
	is true.
	In theorem \ref{ugsPSPACE} we will modify this graph,
	such that it becomes undirected.
	In theorem \ref{dngsPSPACE} we will construct a
	directed overhead graph to $G_\varphi$,
	which will force \alice\ and \bob\ to choose 
	certain starting positions. 
	At last in theorem \ref{ungsPSPACE} we will construct
	an undirected overhead graph. Here we will make
	use of the constructions of the preceding theorems. 
	
	\begin{figure}[h!]
			\begin{center}
				\includegraphics[width = 0.4\textwidth]{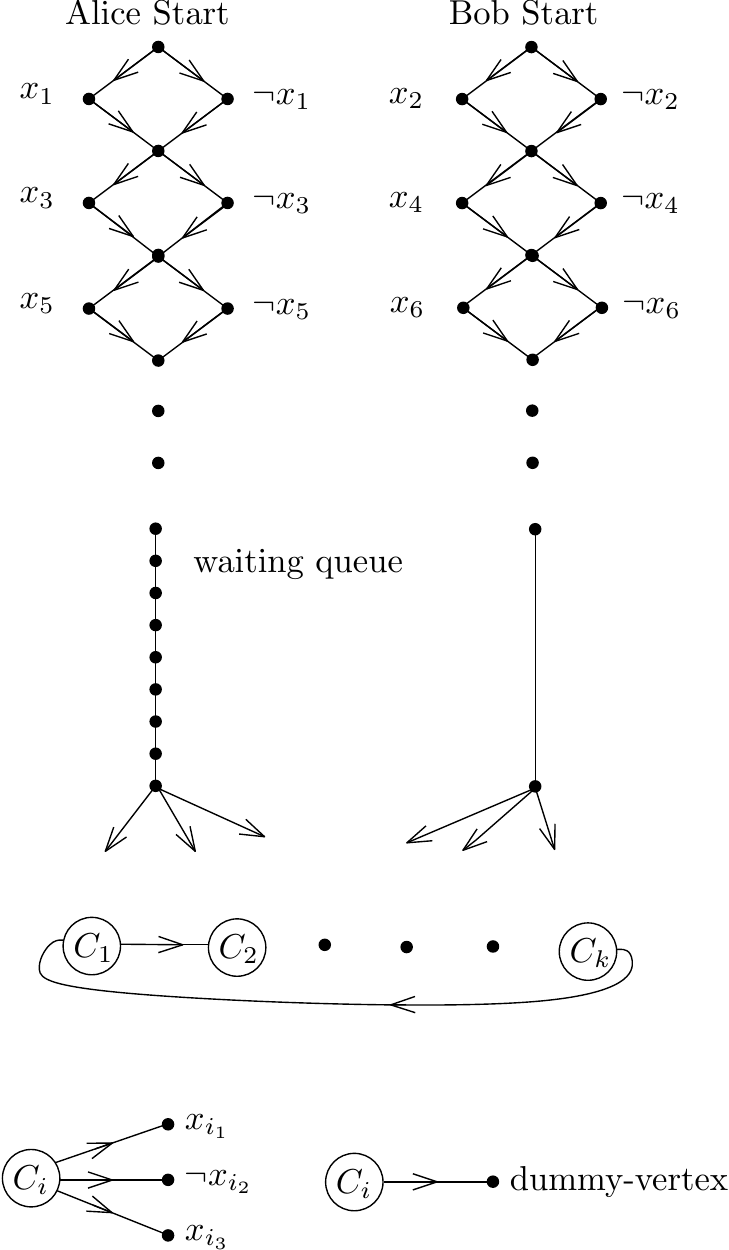}
				\caption{$G_\varphi$}
				\label{fig:directedGivenStartGraph}				
			\end{center}
		\end{figure}
	
	Theorem \ref{dgsPSPACE} has already been  proven 
	by Bodlaender \cite{DBLP:journals/tcs/Bodlaender93}
	and is similar to the proof that generalized geography 
	is PSPACE-complete \cite{DBLP:books/daglib/0086373}.
	We repeat his proof, with subtle changes. These 
	differences are necessary for 
	theorem \ref{ugsPSPACE}, \ref{dngsPSPACE} and 
	\ref{ungsPSPACE}	to work.
	\begin{theorem}\label{dgsPSPACE}
		The problem to decide if \alice\ has a winning strategy in 
		a \emph{directed} graph \emph{with} given start positions
		is PSPACE-complete.
	\end{theorem}
	\begin{proof}
		Given a QBF $\varphi$ with $n$ variables and $k$ clauses
		we construct a directed graph $G_{\varphi}$ as 
		depicted in figure \ref{fig:directedGivenStartGraph}.
		It consists of starting positions for
		\alice\ and for \bob\ from where variable-gadgets 
		begin such that \alice\ and \bob\ 
		have to decide whether they move left or right 
		which represents an assignment of 
		the corresponding variable.
		Thereafter \alice\ has to enter a 
		path of length $k-1$, which we 
		call the waiting queue. Meanwhile
		\bob\ can enter the clause gadget,
		which consists of $k$ vertices arranged
		in a directed cycle each representing exactly one
		clause. 
		Thus \bob\ can traverse all but one clause-vertex
		before \alice\ enters the clause-gadget. 
		When she enters, she has only one clause-vertex
		to go to, which was chosen by \bob .
		Now from each clause-vertex we have edges 
		to the corresponding variables and one edge to  
		a dummy-vertex. So each player can make at most one more turn.
		Thus \bob\ takes the dummy-vertex.
		Consequently if $\varphi$ was true \alice\ had
		a strategy to assign the variables in a way
		that every clause becomes true
		and she is still able to make one more turn and 
	 	therefore wins.
	 	Otherwise \bob\ has a strategy to assign 
	 	the variables, such that at least one clause 
	 	is false. Thus \alice\ cannot move 
	 	anymore from the clause-vertex 
	 	and the game ends in a tie.
	 	This shows PSPACE-hardness.
		As the game ends after a linear number of turns,
		it is possible to traverse the game tree using linear space.
		See \cite{DBLP:books/daglib/0086373} for a similar 
		argument.
	\end{proof}
	
	Our  approach is to take
	the graph $G_{\varphi}$ from theorem \ref{dgsPSPACE} 
	and convert it to a working construction
	for theorem \ref{ugsPSPACE}.

	\begin{theorem}\label{ugsPSPACE}
		The problem to decide whether \alice\ has a winning strategy in 
		an \emph{undirected} graph \emph{with} given start positions
		is PSPACE-complete.
	\end{theorem}
	
	\begin{proof}
		We replace every directed edge of $G_{\varphi}$
		by an undirected one. Further, we will
		carry out $4$ modifications and 
		later prove, that the resulting 
		graph $G_{\varphi}'$ has
		the desired properties.

		\begin{modification}[slow-path]
		As we want that \alice\ and \bob\
		assign each variable in order,
		we must prevent
		them from using the edge from a variable-vertex to 
		a clause-vertex. We achieve this via
		elongating every such edge to a path 
		of  length $2k+n$. See figure \ref{fig:ModWaitingQueueClauseVertex}.
		\end{modification}

		\begin{figure}
			\begin{center}
				\subfloat[Modification of the waiting queue]
				{\includegraphics[width=0.25\textwidth]{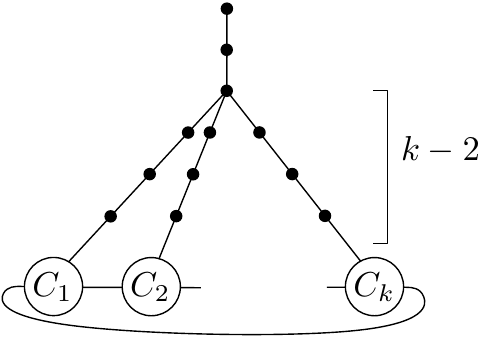}}
				\ \ \ \ \ \ \ 
				\subfloat
				[The way from a clause-vertex to a variable-vertex]
				{\includegraphics[width=0.4\textwidth]{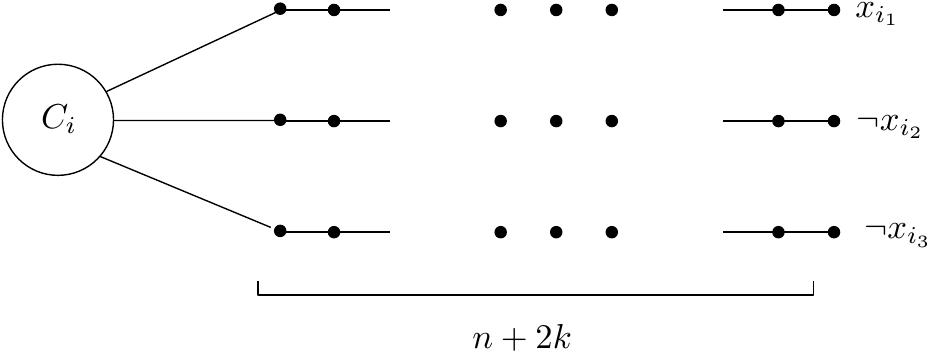}}
				
				\caption{Two of the modifications}
				\label{fig:ModWaitingQueueClauseVertex}
				\end{center}
		\end{figure}
		\begin{modification}[waiting queue]
			The next
			motion that might happen,
			is that \bob\ cuts off the waiting queue.
			We prevent this by replacing
			the waiting queue by the graph
			depicted in figure
			\ref{fig:ModWaitingQueueClauseVertex}.
		\end{modification}
		
		\begin{modification}[dummy-vertex]
			Another concern is
			that \bob\ might go towards the dummy-vertex
			and return. To hinder this we replace all
			the edges to the dummy-vertex by the
			construction in figure \ref{fig:modDummy}.
		\begin{figure}[h!]
				\begin{center}
					\includegraphics[width = 0.4\textwidth]{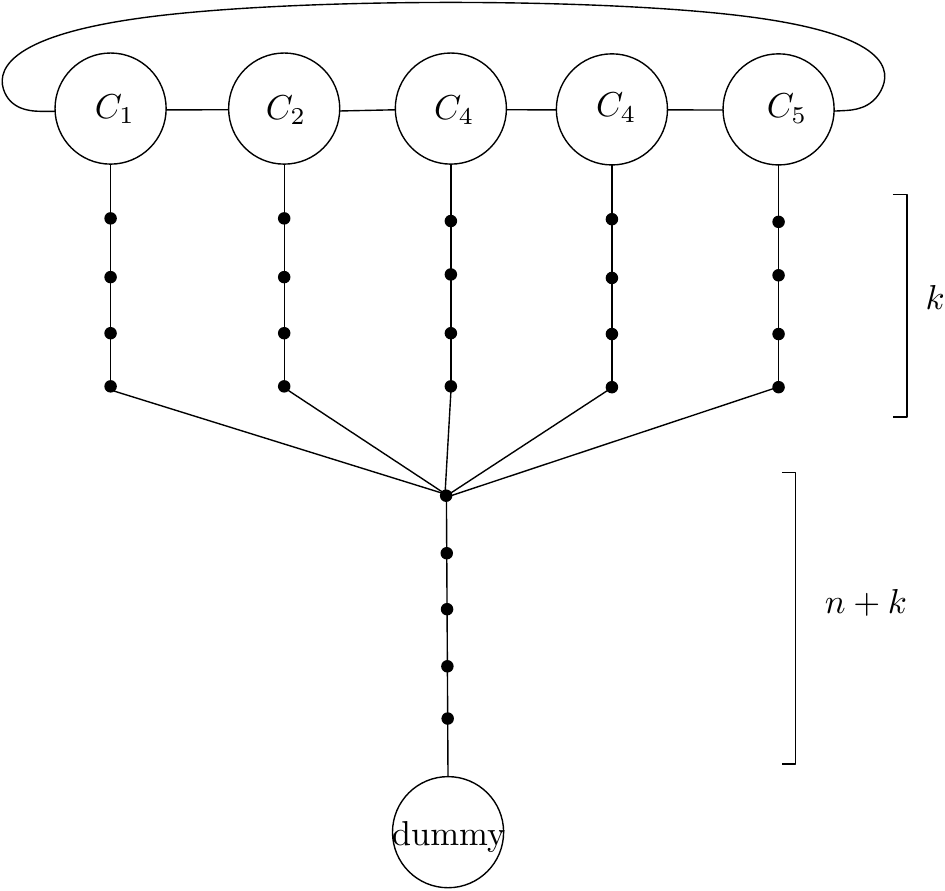}
					\caption{Modification of the paths to the dummy-vertex}
					\label{fig:modDummy}
				\end{center}
		\end{figure}
		\end{modification}

		\begin{modification}[spare-path]
			It might be advantegous for
			\bob\ to go to a literal, which is contained
			in two clauses, instead 
			of going to the dummy-vertex, because he then might
			use the return-path to a clause-gadget 
			and receive in total $4k+2n+1$
			vertices after leaving the clause-gadget. 
			We attach a path of length $2n+k$
			to each variable-vertex and the dummy-vertex.
		\end{modification}
		
		We show first, that after \alice\ and \bob\ leave
		their respective start positions, they have
		to assign the variables.
		There are only two strategies 
		they possibly could follow instead.
		The first is to use a spare-path
		from modification $4$. This gives at most
		$2n+k$ many vertices. The other player would
		just go down to the dummy-vertex and proceed
		to the spare-path from the dummy-vertex.
		Thus using the spare-path at this stage
		leads to a loss.
		The other option is to use a slow-path from a 
		variable to the clause-gadget as introduced in
		modification $1$. It takes quite a while to
		traverse this path and meanwhile, the other
		player can just go down to the
		clause-gadget, traverse all the clause-vertices
		and then go to the dummy-vertex. Again, it turns 
		out that this strategy is not a good option.
		
		So we have established that \bob\
		reaches the clause-gadget, \alice\
		reaches the waiting queue and they have assigned
		all the variables alternatingly on their way.
		Now \bob\ could make one of two plans we would not like.
		The first plan is that he might try to go to the
		dummy-vertex and return before \alice\
		has reached a clause-vertex. But the time
		to return is so long that \alice\ will have taken
		all the clause-vertices meanwhile and
		\bob\ would receive more vertices if he were
		to proceed all the way to the dummy-vertex and take the
		spare-path.
		
		The second plan he might pursue is to
		short-cut the waiting queue. 
		Luckily, the queue splits after $2$
		vertices. So when \bob\ enters the queue before 
		$2$ turns, \alice\ can avoid him by taking 
		a different branch and the planner
		himself gets trapped. If he waits $2$ turns,
		he must have determined a clause-vertex
		for \alice\ already. So \alice\ 
		knows which branch to use.
		This particular branch cannot be reached by
		\bob\ by then. So our constructions have
		circumvented his plans again.
		
		In summary we have established that \bob\  indeed
		has to traverse $k-1$ clause-vertices
		and \alice\ obviously has to go to the clause-vertex
		\bob\ left for her.
		What now?
		It is \bob s turn. One of the longest paths that remains
		goes to the dummy-vertex and proceeds via a spare-path.
		So he had better take it, because otherwise
		\alice\ will take it and he loses.
		
		Now it is \alice\ turn. If there is a variable-vertex
		she can reach, she also has a path of the same length
		as \bob\ does and this would imply that she will win. 
		If not, then she could only go towards a 
		variable-vertex and \bob\ will win.
		
		And again as in theorem \ref{dgsPSPACE} \alice\
		has a winning strategy in $G'_{\varphi}$ if and only if
		$\varphi$ is satisfiable.
	\end{proof}
		
	Now we show how to force \alice\ and \bob\
	to choose certain start positions in
	a \emph{directed} graph.
	We will do that by constructing
	a graph $H(G)$, such that \alice\
	wins in $H(G)$ if and only if 
	\alice\ wins in $G$ when both 
	players start
	at certain positions $v_1$ and $v_2$.
	It follows, that \alice\ wins 
	in $H(G_{\varphi})$ if and only if
	$\varphi$ is true. 
	With a similar but different construction,
	theorem \ref{dngsPSPACE} was shown in 
	\cite{DBLP:journals/tcs/Bodlaender93}.
	Here we give a slightly different proof again,
	because it is an essential step for our proof of
	theorem \ref{ungsPSPACE}.
	\begin{theorem}\label{dngsPSPACE}
		The problem to decide whether \alice\ has a winning strategy in 
		a \emph{directed} graph \emph{without} given start positions
		is PSPACE-complete.
	\end{theorem}
	\begin{proof}
		Assume we are given some directed graph
		$G$ and two vertices $v_1, v_2 \in G$.
		We construct some directed Graph $H(G)$
		such that \alice\ wins in $H(G)$ if and only
		if \alice\ wins in $G$ with the predefined
		start-positions $v_1,v_2$.
		Applied to $G_{\varphi}$, this finishes the proof.
		
		The general idea of such an overhead graph
		is simple. We construct two vertices,
		which are very powerful, so \alice\
		and \bob\ want to start there, but once there,
		they are forced to go to the start-vertices 
		of the original graph. The idea is also used
		in theorem \ref{ungsPSPACE}.
		
		We describe the construction of $H(G)$ as depicted 
		in figure \ref{fig:overDirected} in detail:
		
		\noindent We add two vertices $u_1$ and $u_2$
		with attached directed paths of
		length $2n = 2 \# V(G)$ to the start-vertices 
		$v_1$ and $v_2$  respectively.
		Now the longest path starts at $u_1$
		or $u_2$ and has length between $2n$
		and $3n$.
		Let $l_{low} = 2n$ denote a lower bound on the length $l$
		of the longest path in $G$ and $l_{up} =3n$ an
		upper bound on $l$.

		\begin{figure}[h]
				\begin{center}
					\includegraphics[width = 0.55\textwidth]{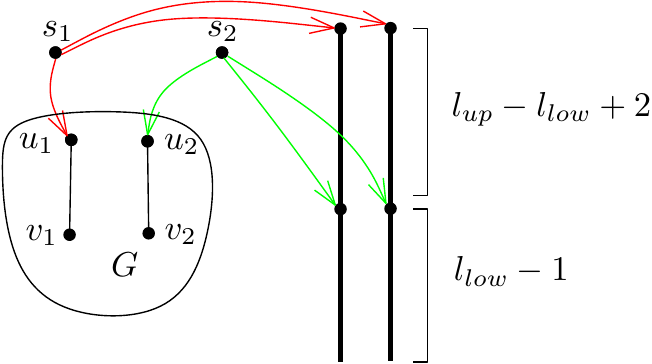}
					\caption{H(G)} 
					\label{fig:overDirected}
				\end{center}
		\end{figure}
		
		Then we add	two directed \emph{auxiliary-paths} 
		of length $l_{up} +1$
		and vertices $s_1$ and $s_2$.
		The vertex $s_1$ is attached to $u_1$
		and the two upper parts of the auxiliary-paths.
		The vertex $s_2$ is connected to
		$u_2$ and the middle parts of auxiliary-paths,
		such that the larger part is below,
		but the larger part is still shorter
		than the longest path from $u_1$ or $u_2$. 
		This is possible as long as 
		$l_{up} - l_{low} + 2 < l_{low}-1$, which is the case.
		We want to point out that
		once one of the players is in an
		auxiliary-path or $G$, there is no way 
		out of the respective component
		simply because there is no outgoing 
		edge.
		
		Assume \alice\ starts at $s_1$
		and \bob\ at $s_2$. Then 
		\alice\ should go to $u_1$
		because otherwise \bob\
		will go into the same auxiliary-path as her
		and receive more 
		vertices than \alice\ and thus she loses.
		Meanwhile, \bob\ should go to $u_2$
		on his first turn, as he would
		receive fewer vertices in an 
		auxiliary-path than \alice\ in $G$.
		Now we show, that
		it is best for \alice\ to start
		at $s_1$. 
		
		\noindent
		\texttt{Case 1}
			\alice\ starts in $G$ Then \bob\ just
			starts at the top of an auxiliary-path.
		
		\noindent
		\texttt{Case 2} 
			\alice\ starts in an auxiliary-path. 
			As the path is directed,
			\bob\ starts in front of her.
			
		\noindent
		\texttt{Case 3}
			\alice\ starts in $s_2$. Then \bob\ starts in $s_1$.
			Now \bob\ can get $l_{up}+2$ in total and \alice\
			at most $l_{up}+1$.
		
		Thus \alice\ is better off starting at $s_1$,
		or she will lose anyway.
		We show now that under these conditions, \bob\
		is always better off starting at $s_2$.
		
		\noindent
		\texttt{Case 4} \bob\ starts in an auxiliary-path.
			\alice\ will go to the other auxiliary-path and win.
		
		\noindent
		\texttt{Case 5} \bob\ starts in $G$. \alice\
		will then just go to an auxiliary-path.
	\end{proof}
	Now the last task is to show the result if
	the graph is \emph{undirected} and the starting
	positions are \emph{not} given.
	We will do that by using the graph 
	$G_\varphi'$ and an undirected version of
	$H(G)$ which we will denote by $H'(G)$. 
	Unfortunately this will not work immediately.
	We will therefore construct an overhead 
	graph $F(H'(G))$ using some properties of
	$H'(G)$.
	
	\begin{theorem}\label{ungsPSPACE}
		The problem to decide whether \alice\ has a winning strategy in 
		a \emph{undirected} graph \emph{without} given start positions
		is PSPACE-complete.
	\end{theorem}
	\begin{proof}
		The general idea of this construction
		is the same as in the previous proof,
		but because we build up from
		the construction from theorem \ref{dngsPSPACE},
		everything gets more involved.
		Every single argument is still elementary.
		
		Let $H'(G)$ be the graph $H(G)$ with
		all directed edges replaced
		by undirected ones. 
		Also the auxiliary paths have to
		be changed slightly, because $l_{low} = 4n$ 
		and $l_{up} = 5n$. 
		We observe $5$ properties of this $H'(G)$:
		\begin{enumerate}
			\item[p1] If \alice\ starts at $s_1$ and \bob\
				starts at $s_2$, then \alice\ has to go 
				to $u_1$ and \bob\ to $u_2$.
			\item[p2] If \alice\ starts at $s_2$ and \bob\
				at $s_1$, \bob\ will win.
			\item[p3] If we assume $s_1$ and $s_2$ are forbidden
				to use, except when started at, it holds that
				the longest path starts at $s_1$.
				(longest path in the sense
				that we consider only one player.)
			\item[p4] Any path from $s_1$ to $s_2$ can
				be extended using an auxiliary-path.
			\item[p5] The shortest path from
			$s_1$ to $s_2$ has length at least $3$.
		\end{enumerate}
		
			Properties p1 and  p2 hold for directed
			graphs according to the proof of theorem \ref{dngsPSPACE},
			and hold by the same arguments
			for the undirected case.
			Property p3 is clear by the definition
			of the auxiliary paths. p4 
			is clear because any path from
			$s_1$ to $s_2$ uses at most one
			auxiliary path. Thus the path can be extended
			to an auxiliary path that has
			not been used yet. To p5 we
			remark that we consider only sufficiently large
			$n$.
		
		\begin{figure}[h!]
				\begin{center}
					\includegraphics[width = 0.4\textwidth]{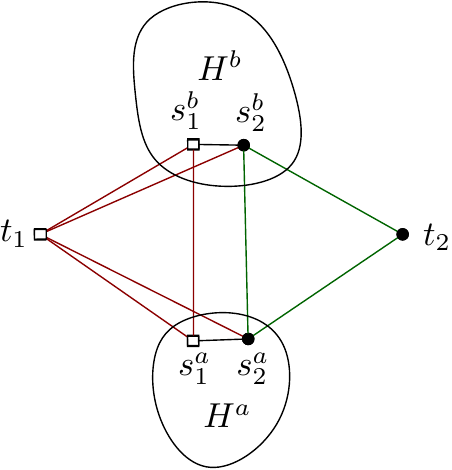}
					\caption{$F(H'(G))$} 
					\label{fig:overUnDirected}
				\end{center}
		\end{figure}
		
		We construct an overhead graph of $H'(G)$, namely
		$F(H'(G))$, as depicted in figure \ref{fig:overUnDirected}.
		It consists of two copies of $H'(G)$, which we call $H^a$
		and $H^b$. In addition two 
		vertices $t_1$ and $t_2$. 
		We indicate with an upper index $^a$ or $^b$ whether a vertex belongs to $H^a$ 
		or $H^b$.
		We will always go w.l.o.g. to $H^a$ instead of $H^b$
		when the situation is symmetric.
		The edge-set consists of all the edges
		in $H^a$, $H^b$ and 
		$(t_1,s^a_1)$,$(t_1,s^b_1)$,$(t_1,s^a_2)$,$(t_1,s^b_2)$,
		$(t_2,s^a_2)$,$(t_2,s^b_2)$,$(s_1^a,s^a_2)$,$(s_1^b,s^b_2)$,$(s^a_1,s^b_1 )$,$(s^a_2,s^b_2 )$.
		
		We call $t_2, s_2^a$ and $s_2^b$ \emph{dot-vertices}
		and $t_1, s_1^a$ and $s_1^b$ \emph{box-vertices}.
		
		First we will show, that if \alice\ wins in $G$ with
		start vertices $v_1$ and $v_2$, then \alice\ will win
		in $F(H'(G))$. 
		This means that we assume that \alice\ has
		a winning strategy in $G$ with 
		the respective start-vertices.
		We give an explicit winning strategy.
		\alice\ starts at $t_1$.

		\noindent
		\texttt{Case 1} 
			\bob\ starts inside $H^a$(i.e. not in $s_1^a$ or $s_2^a$). 
			\bob\ is closer to either $s_1^a$ or $s_2^a$.
			Both can be reached by $t_1$.
			So \alice\ can imprison \bob\ by going
			to the closer vertex and then to the other vertex.
			\bob\ cannot escape, because of p5.
			After that \alice\ can go to $H^b$ and wins there
			by p3.
		
		\noindent
		\texttt{Case 2}
			\bob\ starts at $s^a_1$. \alice\ takes $s^b_1$.
			Then \alice\ copies every move of \bob\ and thus 
			wins, since the only move she cannot copy is
			to $t_2$. But p4 shows us that this is not a wise
			move of \bob .
		
		\noindent	
		\texttt{Case 3} \bob\ starts at $s^a_2$.
			\alice\ goes to $s^a_1$ and $s^b_1$ in this order.
			By then \bob\ is either in 
			$H^a$, where he will lose by
			p3, or he is in $H^b$ and will
			lose by p2, or he will be at $t_2$
			and cannot move, or he is at $s^b_2$
			and will lose by p1 and the assumption.
			
		\noindent
		\texttt{Case 4} \bob\ starts at $t_2$.
			\alice\ will go to $s^a_1$, $s^b_1$
			and then enter $H^b$. \bob\ can either
			enter $H^b$ one turn before \alice\
			and lose by p2. 
			Or he enters $H^b$ one turn after 
			\alice\ and lose by p1 an the assumption.
			Or he enters $H^a$ and loses by p3.

			So far we have shown, that if \alice\ wins in $G$
			she does so in $F(H'(G))$. We will proceed by showing,
			that if \bob\ can achieve at least a tie in $G$, 
			so can he in $F(H'(G))$.
			
		\noindent
		\texttt{Case 5a}	
			\alice\ starts at $t_1$. \bob\ goes to $t_2$.
			Let us say \alice\ goes to $s^a_1$, then \bob\
			will follow her with $s^a_2$. Now if \alice\ enters
			$H^a$, he will as well.
			Otherwise she has to go to
			$s^b_1$. He then enters at $s^b_2$ and gets at least
			a tie. (by assumption and p1)

		\noindent
		\texttt{Case 5b}
			\alice\ starts in $t_1$. \bob\ goes to $t_2$.
			Let us say \alice\ chooses $s^a_2$ as her second
			move. \bob\ can go then to $s^b_2$
			and imitate all of \alice 's moves and thus 
			gets a tie. (see \texttt{Case 2})

		\noindent
		\texttt{Case 6} 
			\alice\ starts inside $H^a$
			(i.e. not $s^a_1$ or $s^b_2$). 
			\bob\ cuts her off and enters
			the other copy via $s^b_1$.
			(p3 and p5, see \texttt{Case 1}
		
		\noindent
		\texttt{Case 7a}
			\alice\ starts at $s^a_1$. Then \bob\ will
			start at $t_1$. Let us say that 
			\alice\ goes to $s^b_1$.
			\bob\ will go to  
			$s^b_2$. Now both have to enter
			$H^b$ and \bob\ acquires at least a tie
			by assumption and p1.

		\noindent
		\texttt{Case 7b} 
			\alice\ starts at $s^a_1$. Then \bob\ will
			start at $t_1$. Let us say that 
			\alice\ goes this time to $s^a_2$.
			\bob\ will than go to $s^b_2$.
			Thus \alice\ has to enter $H^a$ and
			\bob\ can enter $H^b$ via $s^b_1$ and
			thus wins by p3.

		\noindent
		\texttt{Case 8a} 
			\alice\ starts on $s^a_2$. Then \bob\ will
			start at $t_1$. Now if \alice\ goes to 
			$s^a_1$ \bob\ can go to $s^b_1$ and imitate her moves
			as in \texttt{Case 6}. Here he has even more
			options than \alice .

		\noindent
		\texttt{Case 8b} 
			\alice starts at $s^a_2$. Then \bob\ will
			start iat $t_1$. This time we assume
			\alice\ goes to 
			 $s^b_2$, \bob\ takes $s^b_1$. Then \alice\
			can make a last move to $t_2$ or 
			enter $H^b$. In the second case \bob\ goes
			to  $H^a$ via
			$s^a_1$ and wins by p4.
		
		\noindent
		\texttt{Case 9} 
			\alice\ starts at $t_2$.
			\bob\ goes to $t_1$ and follows her
			in the sense that if she goes to
			$s^a_2$, he will go to $s^a_1$.
			Thus either \alice\ enters
		  $H^a$ via $s^a_2$ and \bob\ will enter
		  $H^a$ via $s^a_1$ and thus wins by p2,
		  or the same happens with $H^b$ one turn later.

	\end{proof}

\section*{Acknowledgments}
I want to thank Justin Iwerks for
suggesting this research topic and
initial discussions. For proofreading 
and general advice I want to thank
Wolfgang Mulzer, Lothar Narins and Tobias Keil. 
For the visage on the right-hand 
side of figure \ref{fig:VisageTillTorsten}
I thank Torsten Ueckerdt.

\newpage

\bibliography{TillLib}
\bibliographystyle{plain}
	
\end{document}